\documentclass[a4paper,twoside, 11pt]{journal}
\usepackage{amssymb, amsmath, amsthm, wasysym, mathrsfs}
\usepackage{xspace}
\usepackage{graphicx, wrapfig}
\usepackage{tabularx}
\usepackage[colorlinks]{hyperref}
\usepackage[format=plain,indention=1em]{caption}
\usepackage{booktabs}
\usepackage{scalerel}
\usepackage{comment}
\usepackage[usenames,dvipsnames]{xcolor}

\usepackage{times}
\usepackage[scaled]{helvet}
\usepackage{microtype}

\usepackage[left=1in,top=1in,right=1in,bottom=1in]{geometry}
\graphicspath{{figures/}}

\hypersetup{allcolors = violet}

\author{
  Lars Arge\thanks{
    MADALGO, Aarhus University,
    \texttt{[large,f.staals]@cs.au.dk}
  }
  \and
  Frank Staals\footnotemark[1]
}
\title{Dynamic Geodesic Nearest Neighbor Searching in a Simple Polygon}

\newtheorem{theorem} {Theorem}
\newtheorem{lemma}[theorem]{Lemma}
\newtheorem{corollary}[theorem] {Corollary}
\newtheorem{observation}[theorem]{\textbf{Observation}}

\DeclareFontFamily{OT1}{pzc}{}
\DeclareFontShape{OT1}{pzc}{m}{it}{<-> s * [0.900] pzcmi7t}{}
\DeclareMathAlphabet{\mathpzc}{OT1}{pzc}{m}{it}

\newcommand{\mkmcal}[1]{\ensuremath{\mathcal{#1}}\xspace}

\newcommand{\VD}{\mkmcal{V}}
\newcommand{\F}{\mkmcal{F}}

\renewcommand{\P}{\mkmcal{P}}

\newcommand{\mkmbb}[1]{\ensuremath{\mathbb{#1}}\xspace}

\newcommand{\R}{\mkmbb{R}}

\newcommand{\eps}{\varepsilon}

\newcommand{\etal}{et al.\xspace}

\newcommand{\geod}{\scaleobj{1.15}{\varsigma}\xspace}
\newcommand{\geodlen}{\scaleobj{1.7}{\varsigma}\xspace}

\newcommand{\acall}[2]{\textsc{#1}(#2)\xspace}

\makeatletter
\renewcommand*{\@fnsymbol}[1]{\ensuremath{\ifcase#1\or *\or \mathsection\or \mathparagraph\or
    \dagger\or \ddagger\or \|\or **\or \dagger\dagger
    \or \ddagger\ddagger \else\@ctrerr\fi}}
\makeatother

\begin{document}
\maketitle

\begin{abstract}
  We present an efficient dynamic data structure that supports geodesic nearest
  neighbor queries for a set of point sites $S$ in a static simple polygon
  $P$. Our data structure allows us to insert a new site in $S$, delete a site
  from $S$, and ask for the site in $S$ closest to an arbitrary query point
  $q \in P$. All distances are measured using the geodesic distance, that is,
  the length of the shortest path that is completely contained in $P$. Our data
  structure supports queries in $O(\sqrt{n}\log n\log^2 m)$ time, where $n$ is
  the number of sites currently in $S$, and $m$ is the number of vertices of
  $P$, and updates in $O(\sqrt{n}\log^3 m)$ time. The space usage is
  $O(n\log m + m)$. If only insertions are allowed, we can support queries in
  worst-case $O(\log^2 n\log^2 m)$ time, while allowing for $O(\log n\log^3 m)$
  amortized time insertions. We can achieve the same running times in case
  there are both insertions and deletions, but the order of these operations is
  known in advance.
\end{abstract}


\thispagestyle{empty}
\clearpage
\setcounter{page}{1}

\section{Introduction}
\label{sec:Introduction}

Nearest neighbor searching is a classic problem in computational geometry in
which we are given a set of point \emph{sites} $S$, and we wish to preprocess these
points such that for a query point $q$, we can efficiently find the site
$s \in S$ closest to $q$. 
We consider the case where $S$ is a \emph{dynamic} set of points inside a
simple polygon $P$. That is, we may insert a new site into $S$ or delete an
existing one. We measure the distance between two points $p$ and $q$ by their
\emph{geodesic distance} $\geodlen(p,q)$: the length of the \emph{geodesic}
$\geod(p,q)$. The geodesic $\geod(p,q)$ is the shortest path connecting $p$ and
$q$ that is completely contained in $P$. 

\paragraph{Motivation.} Our motivation for studying dynamic geodesic nearest
neighbor searching originates from a problem in
Ecology~\cite{burrows2014geographical,ordonez2013climatic}. We are given a
threshold $\eps$, and two sets of points in $\R^2$: a set of ``red'' points
$R$, representing the locations at which an animal or plant species lived many
years ago, and and a set of ``blue'' points $B$, representing 
locations where the species could occur today. Each point $p \in R\cup B$ also
has a real value $p_v$, representing an environmental value such as
temperature. The problem is to find, for every species (red point), the closest
current location (blue point) where it can migrate to, provided that the
environmental value (temperature) is similar to its original location,
i.e.~differs by at most $\eps$.

In the setting described above, it is easy to solve the problem in
$O(n\log^2 n)$ time, where $n$ is the total size of $R$ and $B$. Simply build a
balanced binary search tree that stores the blue points in its leaves, ordered
by their $v$-values, and associate each internal node with the Voronoi diagram
of its descendants. For each red point $r$ we can then find the closest blue
point in $O(\log^2 n)$ time by selecting the nodes that together represent the
interval $[r_v-\eps,r_v+\eps]$, and locating the closest point in each
associated Voronoi diagram. However, the geographical environment may limit
migration. For example, if the species considered is a land-based animal like a
deer then it cannot cross a large water body. Hence, we would like to consider
the problem in a more realistic environment. We restrict the movement of the
species to a simple polygon $P$ modeling the land, and measure distances using
the geodesic distance.

Directly applying the previous approach in this new setting, this time building
geodesic Voronoi diagrams, incurs a cost proportional to the size of the
polygon, $m$, in every node of the tree. Thus this approach has a running time
of $\Omega(nm)$. If, instead, we sweep a window of width $2\eps$ over the range
of values, while maintaining our (offline) geodesic nearest neighbor data
structure storing the set of blue points whose value lies in the window, we can
solve the problem in only $O(n(\log^2 n\log^2 m + \log^3 m) + m)$ time. This is
a significant improvement over the previous method.

\paragraph{Related Work.} A well known solution for nearest neighbor searching
with a fixed set of $n$ sites in $\R^2$ is to build the Voronoi diagram and
preprocess it for planar point location. This yields an optimal solution that
allows for $O(\log n)$ query time using $O(n)$ space and $O(n \log n)$
preprocessing time~\cite{bkos2008}. 
Voronoi diagrams have also been studied in case the set of sites is restricted
to lie in a polygon $P$, and we measure the distance between two points $p$ and
$q$ by their geodesic distance
$\geodlen(p,q)$. Aronov~\cite{aronov1989geodesic} shows that when $P$ is simple
and has $m$ vertices, the geodesic Voronoi diagram has complexity $O(n+m)$ and
can be computed in $O((n+m)\log(n+m)\log n)$ time. Papadopoulou and
Lee~\cite{papadopoulou1998geodesic} present an improved algorithm that runs in
$O((n+m)\log(n+m))$ time. Hershberger and Suri~\cite{hershberger1999sssp} give
an $O(m \log m)$ time implementation of the \emph{continuous dijkstra}
technique for the construction of a \emph{shortest path map}. The shortest path
map supports $O(\log m)$ time geodesic distance queries between a fixed source
point $s$ and an arbitrary query point $q$, even in a polygon with holes. When
running their algorithm ``simultaneously'' on all source points (sites) in $S$,
their algorithm constructs the geodesic Voronoi diagram, even in a polygon with
holes, in $O((n+m)\log(n+m))$ time. These results all allow for $O(\log(n+m))$
time nearest neighbor queries. Unfortunately, these results are efficient only
when the set of sites $S$ is fixed, as inserting or deleting even a single site
may cause a linear number of changes in the Voronoi diagram.

To support nearest neighbor queries, it is, however, not necessary to
explicitly maintain the (geodesic) Voronoi diagram. Bentley and
Saxe~\cite{bentley1980decomposable} show that nearest neighbor searching is a
\emph{decomposable search problem}. That is, we can find the answer to a query
by splitting $S$ into groups, computing the solution for each group
individually, and taking the solution that is best over all groups. This
observation has been used in several other approaches for nearest neighbor
searching with the Euclidean
distance~\cite{agarwal1995dynamic,dobkin1991maintenance}.\footnote{Indeed, it
  is also used in our initial solution to the migration problem.} However, even
with this observation, it is hard to get both polylogarithmic update and query
time. Only recently, Chan~\cite{chan2010dynamic_ch} managed to achieve
such results by maintaining the convex hull of a set of points in $\R^3$. Via a
well-known lifting transformation this also allows (Euclidean) nearest neighbor
queries for points in $\R^2$. Chan's solution uses $O(n)$ space, and allows for
$O(\log^2 n)$ queries, while supporting insertions and deletions in
$O(\log^3 n)$ and $O(\log^6 n)$ amortized time, respectively. Very recently,
Kaplan \etal~\cite{dynamic2017kaplan} managed to reduce
the deletion time to $O(\log^5 n)$. In addition, they obtain polylogarithmic
update and query times for more general, constant complexity, distance
functions. Note however that the function describing the geodesic distance may
have complexity $\Theta(m)$, and thus these results do not transfer easily to
our setting.

In the geodesic case, directly combining the decomposable search problem
approach with the static geodesic Voronoi diagrams described above does not
lead to an efficient solution. Similar to in our migration problem, this leads
to an $\Omega(m)$ cost corresponding to the complexity of the polygon on every
update. Simultaneously and independently from us Oh and
Ahn~\cite{oh_ahn2017voronoi} developed an approach that answers queries in
$O(\sqrt{n}\log(n+m))$ time, and updates in $O(\sqrt{n}\log n\log^2 m)$
time. Some of their ideas are similar to ours.




\paragraph{Our Results.} We develop a dynamic data structure to support nearest
neighbor queries for a set of sites $S$ inside a (static) simple polygon $P$. Our
data structure allows us to locate the site in $S$ closest to a query point
$q \in P$, to insert a new site $s$ into $S$, and to delete a site from
$S$. 
Our data structure supports queries in $O(\sqrt{n}\log n\log^2 m)$ time, and
updates in $O(\sqrt{n}\log^3 m)$ time, where $n$ is the number of sites
currently in $S$ and $m$ is the number of vertices of $P$. The space usage is
$O(n\log m + m)$.

As with other decomposable search problems~\cite{bentley1980decomposable}, we
can adapt our data structure to improve the query and update time if there are
no deletions. In this insertion-only setting, queries take worst-case
$O(\log^2 n\log^2 m)$ time, and insertions take amortized $O(\log n\log^3 m)$
time. Furthermore, we show that we can achieve the same running times in case
there are both insertions and deletions, but the order of these operations is
known in advance. The space usage of this version is $O(n\log n\log m + m)$.

\begin{figure}[tb]
  \centering
  \includegraphics{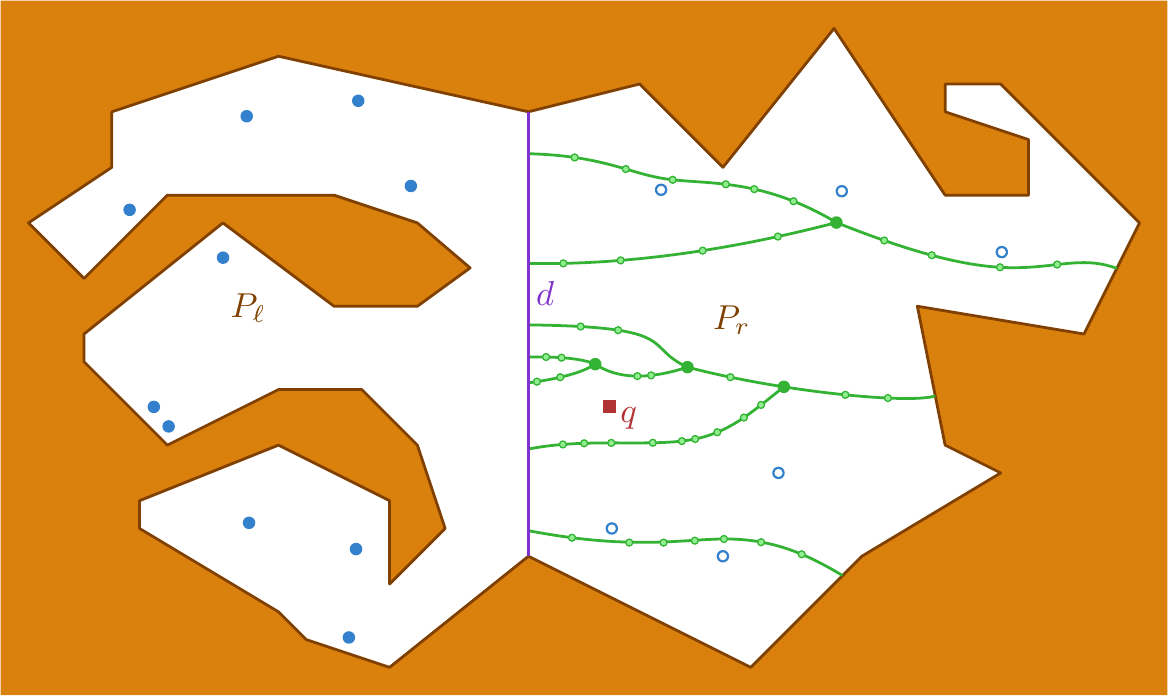}
  \caption{A sketch of the global approach. The diagonal $d$ splits $P$ into
    $P_\ell$ and $P_r$. The $k$ sites in $P_\ell$ induce a forest $\VD$ in
    $P_r$ with $O(k)$ degree three nodes, but potentially $O(m)$ degree two
    nodes. Using \VD we can locate the site closest to $q$ among the sites in
    $P_\ell$.}
  \label{fig:global_approach}
\end{figure}

\paragraph{The Global Approach.} The general idea in our approach is to
recursively partition the polygon into two roughly equal size sub-polygons
$P_\ell$ and $P_r$ that are separated by a diagonal. Additionally, we partition
the sites that lie in $P_\ell$ into a small number of subsets. The Voronoi
diagram that such a subset $S^*$ induces in the \emph{other} half of the
polygon, $P_r$, is a forest \VD. See Fig.~\ref{fig:global_approach} for an
illustration. We show that we can efficiently construct a compact
representation of this forest that supports planar point location queries. 
We handle the sites that lie in $P_r$ analogously. When we get a query point
$q \in P_r$, our forests allow us to find the site in $P_\ell$ closest to $q$
efficiently. To find the site in $P_r$ closest to $q$, we recursively query in
sub-polygon $P_r$. When we add or remove a site $s$ we have to rebuild the
forests associated with only few subsets containing $s$. We show that we can
recompute each forest \VD efficiently.

We give a more detailed description of the approach in
Section~\ref{sec:Global_Approach}. The core of our solution is that we can
represent, and construct, the Voronoi diagram \VD that a set of sites $S^*$ in
$P_\ell$ induces in $P_r$ in time proportional to the size $k$ of the subset
$S^*$. In particular, our representation has a size $O(k)$ and can be built in
time $O(k \log^2 m)$. We show this in Section~\ref{sec:rebuilding}. The key to
achieving this result is a representation of the bisector $b_{st}$ of two sites
$s$ and $t$. Our representation of $b_{st}$, presented in
Section~\ref{sec:Bisector}, allows us to find the intersection of $b_{st}$ with
another bisector $b_{tu}$ in $P_r$ efficiently, and can be obtained from the input polygon in
$O(\log^2 m)$ time.  We combine all of the components into a fully dynamic data
structure in Section~\ref{sec:Together}. Furthermore, we show that we can get
improved query and update times in the insertion-only and
offline-cases.

\section{An Overview of the Data Structure}
\label{sec:Global_Approach}

As in previous work on geodesic Voronoi diagrams~\cite{aronov1989geodesic,
  papadopoulou1998geodesic}, we assume that $P$ and $S$ are in general
position. That is, (i) no
two sites $s$ and $t$ in $S$ (ever) have the same geodesic distance to a vertex
of $P$, and (ii) no three points (either sites or vertices) are colinear. Note
that (i) implies that no bisector $b_{st}$ between sites $s$ and $t$ contains
a vertex of $P$.


We start by preprocessing $P$ for two-point shortest path queries using the
data structure by Guibas and Hershberger~\cite{guibas1989query} (see also the
follow up note of Hershberger~\cite{hershberger_new_1991}). This takes $O(m)$
time and allows us to compute the geodesic distance $\geodlen(p,q)$ between any
pair of query points $p,q \in P$ in $O(\log m)$ time. We then construct a
balanced decomposition of $P$ into sub-polygons~\cite{guibas1987balanced}. A
balanced decomposition is a binary tree in which each node represents a
sub-polygon $P'$ of $P$, together with a diagonal $d$ of $P'$ that splits $P'$
into two sub-polygons $P_\ell$ and $P_r$ that have roughly the same number of
vertices. As a result, the height of the tree, and thus the number of
\emph{levels} in the decomposition, is $O(\log m)$. The root of the tree
represents the polygon $P$ itself.

Consider a diagonal $d$ of $P'$ that splits $P'$ into $P_\ell$ and $P_r$, and
let $S_\ell = S \cap P_\ell$ denote the set of sites in $P_\ell$. The Voronoi
diagram of $S_\ell$ in $P_r$ is a forest with $O(k)$, where $k=|S_\ell|$, nodes
of degree one or three, and $O(m)$ nodes of degree
two~\cite{aronov1989geodesic}. The degree three nodes correspond to
intersection points of two bisectors, and the degree one nodes correspond to
intersection points of a bisector with the polygon boundary. See
Fig.~\ref{fig:global_approach} for an illustration. We refer to the topological
structure of the forest, i.e.~the forest with only the degree one and three
nodes, of $S_\ell$ as $\VD$.

At every level of the decomposition, we partition the sites in $S_\ell$ in
$O(\sqrt{n})$ subsets, each of size $O(\sqrt{n})$. For each subset $S^*$ we
explicitly build the embedded forest $\VD$ that represents its Voronoi diagram
in $P_r$ using an algorithm for constructing a Hamiltonian abstract Voronoi
diagram~\cite{klein1994hamiltonian_vd}. More specifically, for every degree one
or three node, we compute the location of the intersection point that it is
representing, and to which other nodes (of degree one or three) it is connected
to. Furthermore, we preprocess $\VD$ for planar point location. Note that the
edges in this planar subdivision correspond to pieces of bisectors, and thus
are actually chains of hyperbolic arcs, each of which may have a high internal
complexity. We do not explicitly construct these chains, but show that there is
an oracle that can decide if a query point $q \in P_r$ lies above or below a
chain (edge of the planar subdivision) in $O(\log m)$ time. It follows that our
representation of $\VD$ has size $O(k)$, where $k=|S^*|$, and supports planar
point location queries in $O(\log k\log m)$ time. We handle the sites
$S_r = S \cap P_r$ in $P_r$ analogously.

Once $P'$ is a triangle, corresponding to a leaf in the balanced decomposition,
the geodesic between any pair of points in $P'$ is a single line segment, and
thus the geodesic distance equals the Euclidean distance. In this case, we
maintain the sites in $P'$ in a dynamic Euclidean nearest neighbor data
structure such as the one of Chan~\cite{chan2010dynamic_ch} or the much simpler
data structure of Bentley and Saxe~\cite{bentley1980decomposable}.

Since every site $s$ is stored in exactly one subset at every level of the
decomposition, and the Voronoi diagram for each such subset has linear
size, the data structure uses $O(n\log m + m)$ space.

\paragraph{Handling a Query.} Consider a nearest neighbor query with point
$q \in P'$ at a node of the balanced decomposition corresponding to sub-polygon
$P'$, and let $P_\ell$ and $P_r$ be the sub-polygons into which $P'$ is
split. In case that $q \in P_r$, we find the site $s$ in $P_\ell$
closest to $q$, and recursively query the data structure for the nearest
neighbor of $q$ in sub-polygon $P_r$. We handle the case that $q \in P_\ell$
analogously. Once $P'$ is a triangle we find the site closest to $q$ by
querying the Euclidean nearest neighbor searching data structure associated
with $P'$. For each of the $O(\log m)$ levels of the balanced decomposition
this gives us a candidate closest site, and we return the one that is closest
over all.

Since we partitioned the set of sites $S_\ell$ that lie in $P_\ell$ into
$O(\sqrt{n})$ subsets, we can find the site $s \in S_\ell$ closest to $q$ by a
point location query in the Voronoi diagram associated with each of these
subsets. Each such a query takes $O(\log n\log m)$ time, and thus we can find
$s$ in $O(\sqrt{n}\log n\log m)$. The final query in the Euclidean nearest
neighbor data structure can easily be handled in $O(\sqrt{n}\log n)$
time~\cite{bentley1980decomposable}. Since we have $O(\log m)$ levels in the
balanced decomposition, the total query time is $O(\sqrt{n}\log n\log^2 m)$.

\paragraph{Handling Updates.} Consider inserting a new site $s$ into the
data structure, or removing $s$ from $S$. The site $s$ needs to be, or is, stored
in exactly one subset at every level of the decomposition. Suppose that $s$
needs to be, or is, in the subset $S^*$ of $S_\ell$ at some level of the
decomposition. We then simply rebuild the forest \VD associated with $S^*$. In
Section~\ref{sec:rebuilding} we will show that we can do this in
$O(|\VD|\log^2 m)$ time. Since the subset containing $s$, and thus its
corresponding forest \VD, has size $O(\sqrt{n})$, the cost per level is
$O(\sqrt{n}\log^2 m)$. Inserting $s$ into or deleting $s$ from the final
Euclidean nearest neighbor data structure can be done in $O(\sqrt{n})$
time~\cite{bentley1980decomposable}. It follows that the total update time is
$O(\sqrt{n}\log^3 m)$.

\section{Representing a Bisector}
\label{sec:Bisector}

\newcommand{\br}{b^*}

Assume without loss of generality that the diagonal $d$ that splits $P$ into
$P_\ell$ and $P_r$ is a vertical line-segment, and let $s$ and $t$ be two sites
in $S_\ell$. In this section we show that there is a representation of
$\br_{st} = b_{st} \cap P_r$, the part of the bisector $b_{st}$ that lies in
$P_r$, that allows efficient random access to the bisector vertices. Moreover,
we can obtain such a representation using a slightly modified version of the
two-point shortest path data structure of Guibas and
Hershberger~\cite{guibas1989query}.

Let $s$ be a site in $S_\ell$, and consider the shortest path tree $T$ rooted
at $s$. Let $e=\overline{uv}$ be an edge of $T$ for which $v$ is further away from $s$
than $u$. The half-line starting at $v$ that is colinear with, and extending
$e$ has its first intersection with the boundary $\partial P$ of $P$ in a point $w$. We refer to
the segment $\overline{vw}$ as the \emph{extension segment} of
$v$~\cite{aronov1989geodesic}. Let $E_s$ denote the set of all extension
segments of all vertices in $T$.

\begin{figure}[tb]
  \centering
  \includegraphics[page=2]{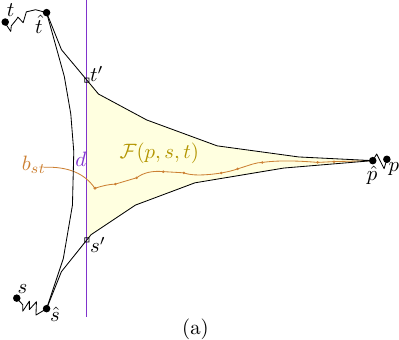}
  \quad
  \includegraphics[page=3]{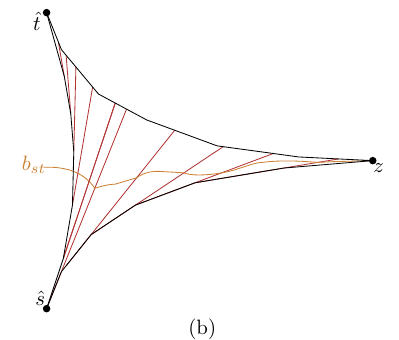}
  \caption{(a) The polygon $\P(p,s,t)$ bounded by the shortest paths between
    $s$, $p$, and $t$ is a pseudo-triangle $\hat{\P}(p,s,t)$ with polylines
    attached to its corners $\hat{s}$, $\hat{p}$, and $\hat{t}$. It contains
    the funnel $\F(p,s,t)$. (b) The clipped extension segments in $F_s^t$ are all
    pairwise disjoint, and end at the chain from $t$ to $z$.}
  \label{fig:bisector_in_funnel}
\end{figure}

Consider two sites $s,t \in S_\ell$, and its bisector $b_{st}$. We then have

\begin{lemma}[Lemma 3.22 of Aronov~\cite{aronov1989geodesic}]
  \label{lem:bisector_prop}
  The bisector $b_{st}$ is a smooth curve connecting two points on $\partial P$
  and having no other points in common with $\partial P$. It is the
  concatenation of $O(m)$ straight and hyperbolic arcs. The points along
  $b_{st}$ where adjacent pairs of these arcs meet, i.e.,~the vertices of
  $b_{st}$, are exactly the intersections of $b_{st}$ with the segments of $E_s$ or
  $E_t$.
\end{lemma}

\begin{lemma}[Lemma~3.28 of Aronov~\cite{aronov1989geodesic}]
  \label{lem:bisector_sp_intersect}
  For any point $p \in P$, the bisector $b_{st}$ intersects the shortest path
  $\geod(s,p)$ in at most a single point.
\end{lemma}

\begin{figure}[tb]
  \centering
  \includegraphics{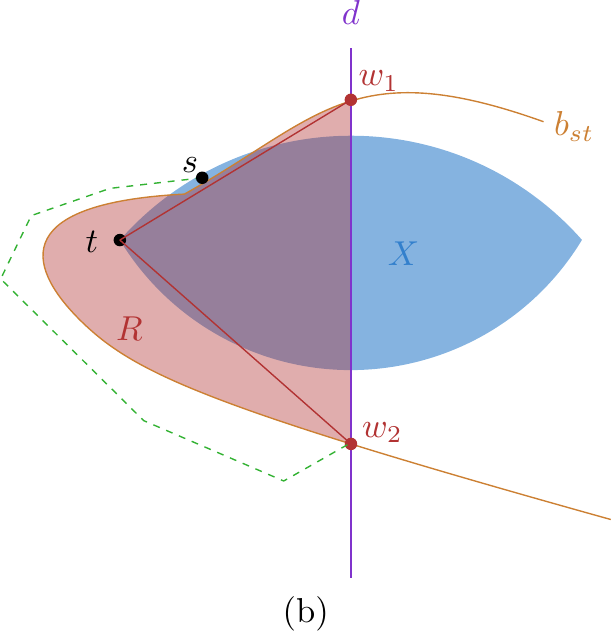}
  \caption{The
    geodesic distance from $t$ to $w_1$ and $w_2$ equals its Euclidean
    distance. The shortest path from $s$ to $w_2$ (dashed, green) has to go
    around $R$, and is thus strictly longer than $\|tw_2\|$.}
  \label{fig:intersection_diagonal}
\end{figure}


Consider a point $p$ on $\partial P_r$ and let $\P(p,s,t)$ be the polygon
defined by the shortest paths $\geod(s,p)$, $\geod(p,t)$, and
$\geod(t,s)$. This polygon $\P(p,s,t)$ is a pseudo-triangle $\hat{\P}(p,s,t)$
whose corners $\hat{s}$, $\hat{t}$, and $\hat{p}$, are connected to $s$, $t$,
and $p$ respectively, by arbitrary polylines.

Let $s'$ and $t'$ be the intersection points between $d$ and the geodesics
$\geod(p,s)$ and $\geod(p,t)$, respectively, and assume without loss of generality that $s'_y
\leq t'_y$. The restriction of $\P(p,s,t)$ to $P_r$ is a
\emph{funnel} $\F(p,s,t)$, bounded by $\geod(t',p)$, $\geod(p,s')$, and
$\overline{s't'}$. See Fig.~\ref{fig:bisector_in_funnel}(a). Note that 
$\geod(s,t)$ is contained in $P_\ell$.

Clearly, if $b_{st}$ intersects $P_r$ then it intersects $d$. There is at most one such intersection point:


\begin{lemma}
  \label{lem:bst_intersect_d}
  The bisector $b_{st}$ intersects $d$ in at most one point $w$.
\end{lemma}

\begin{proof}
  Assume, by contradiction, that $b_{st}$ intersects $d$ in two points $w_1$
  and $w_2$, with $w_1$ above $w_2$. See
  Fig.~\ref{fig:intersection_diagonal}(b). Note that by
  Lemma~\ref{lem:bisector_prop}, $b_{st}$ cannot intersect $\partial P_\ell$,
  and thus $d$, in more than two points. Thus, the part of $b_{st}$ that lies
  in $P_\ell$ between $w_1$ and $w_2$ does not intersect $\partial
  P_\ell$. Observe that this implies that the region $R$ enclosed by this part
  of the curve, and the part of the diagonal from $w_1$ to $w_2$
  (i.e. $\overline{w_1w_2}$) is empty. Moreover, since the shortest paths from
  $t$ to $w_1$ and to $w_2$ intersect $b_{st}$ only once
  (Lemma~\ref{lem:bisector_sp_intersect}) region $R$ contains the shortest
  paths $\geod(t,w_1)=\overline{tw_1}$ and $\geod(t,w_2)=\overline{tw_2}$.

  Since $s$ has the same geodesic distance to $w_1$ and $w_2$ as $t$, $s$ must
  lie in the intersection $X$ of the disks $D_i$ with radius $\|tw_i\|$
  centered at $w_i$, for $i \in 1,2$. It now follows that $s$ lies in one of
  the connected sets, or ``pockets'', of $X \setminus R$. Assume without loss
  of generality that it lies in a pocket above $t$ (i.e.~$s_y > t_y$). See
  Fig.~\ref{fig:intersection_diagonal}. We now again use
  Lemma~\ref{lem:bisector_sp_intersect}, and get that $\geod(s,w_2)$ intersects
  $b_{st}$ only once, namely in $w_2$. It follows that the shortest path from
  $s$ to $w_2$ has to go around $R \ni t$, and thus has length strictly larger
  than $\|tw_2\|$. Contradiction.
\end{proof}

Since $b_{st}$ intersects $d$ only once (Lemma~\ref{lem:bst_intersect_d}), and
there is a point of $b_{st}$ on $\geod(s,t) \subset P_\ell$, it follows that
there is at most one point $z$ where $b_{st}$ intersects $\partial P_r$ the
outer boundary of $P_r$, i.e.~$\partial P_r \setminus d$. Observe that
therefore $z$ is a corner of the pseudo-triangle $\hat{\P}(z,s,t)$, and that
$\F(z,s,t) \subseteq \hat{\P}(z,s,t)$. Let $\br_{st}=b_{st} \cap P_r$ and
orient it from $w$ to $z$. We assign $b_{st}$ the same orientation.

\begin{lemma}
  \label{lem:bisector_in_funnel}
  (i) The bisector $b_{st}$ does not intersect $\geod(s,z)$ or $\geod(t,z)$ in
  any point other than $z$. (ii) The part of the bisector $b_{st}$ that lies in
  $P_r$ is contained in $\F(z,s,t)$.
\end{lemma}

\begin{proof}
  By Lemma~\ref{lem:bisector_sp_intersect} the shortest path from $s$ to any
  point $v \in P$, so in particular to $z$, intersects $b_{st}$ in at most one
  point. Since, by definition, $z$ lies on $b_{st}$, the shortest path
  $\geod(s,z)$ does not intersect $b_{st}$ in any other point. The same applies
  for $\geod(t,z)$, thus proving (i). For (ii) we observe that any internal
  point of $\geod(s,z)$ is closer to $s$ than to $t$, and any internal point of
  $\geod(t,z)$ closer to $t$ than to $s$. Thus, $\geod(s,z)$ and $\geod(t,z)$
  must be separated by $b_{st}$. It follows that $b_{st} \cap P_r$ lies inside
  $\F(z,s,t)$.
\end{proof}

\begin{lemma}
  \label{lem:extension_segments}
  All vertices of $\br_{st}$ lie on extension segments of the vertices
  in the pseudo-triangle $\hat{\P}(z,s,t)$.
\end{lemma}

\begin{proof}
  Assume by contradiction that $v \neq w$ is a vertex of
  $\br_{st}=b_{st} \cap P_r$ that is not defined by an extension segment of a
  vertex in $\hat{\P}(z,s,t)$. Instead, let $e \in E_s$ be the extension
  segment containing $v$, and let $u \in P \setminus \hat{\P}(z,s,t)$ be the
  starting vertex of $e$. So $\geod(s,v)$ has $u$ as its last
  internal vertex.

  \begin{figure}[tb]
    \centering
    \includegraphics[page=4]{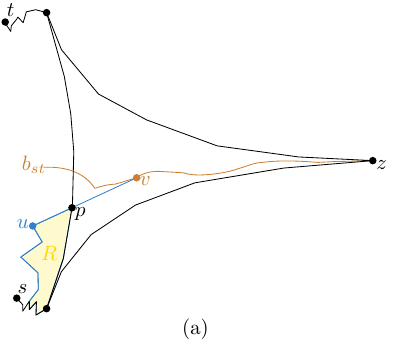}
    \quad
    \includegraphics[page=3]{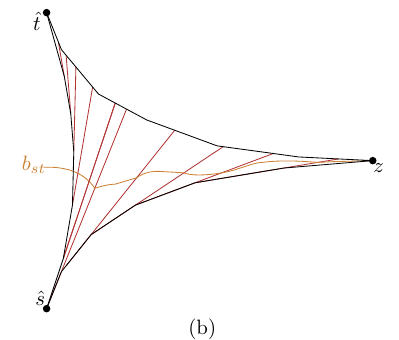}
    \caption{(a) Point $v$ lies inside $\hat{\P}(z,s,t)$, so a shortest path
      from $s$ to $v$ that uses a vertex $u$ outside of $\hat{\P}(z,s,t)$
      intersects $\partial\hat{\P}(z,s,t)$ in a point $p$. This either
      yields two distinct shortest paths from $s$ to $p$, or requires the
      shortest path from $s$ to $p$ via $u$ to intersect $b_{st}$ twice. Both
      yield a contradiction. (b) The extension segments in $F_s^t$ are all pairwise disjoint,
    and end at the chain from $t$ to $z$.}
    \label{fig:extension_segments}
  \end{figure}

  By Lemma~\ref{lem:bisector_in_funnel}, $\br_{st}$ is contained in $\F(z,s,t)$
  and thus in $\hat{\P}(z,s,t)$. Hence, $v \in \hat{\P}(z,s,t)$. Since
  $v \in \hat{\P}(z,s,t)$, and $u \not\in \hat{\P}(z,s,t)$ the shortest path from
  $s$ to $v$ intersects $\partial \hat{\P}(z,s,t)$ in some point $p$. See
  Fig.~\ref{fig:extension_segments}(a). We then distinguish two cases: either $p$
  lies on $\geod(s,z) \cup \geod(s,t)$, or $p$ lies on $\geod(t,z)$.

  In the former case this means there are two distinct shortest paths between
  $s$ and $p$, that bound a region $R$ that is non-empty, that is, it has
  positive area. Note that this region exists, even if $u$ lies on the shortest
  path from $s$ to its corresponding corner $\hat{s}$ in $\hat{\P}(z,s,t)$ but
  not on $\hat{\P}(z,s,t)$ itself (i.e.
  $u \in \geod(s,t)\cup\geod(s,z) \setminus \hat{\P}(z,s,t)$. Since $P$ is a
  simple polygon, this region $R$ is empty of obstacles, and we can shortcut
  one of the paths to $p$. This contradicts that such a path is a shortest
  path.

  In the latter case the point $p$ lies on $\geod(t,z)$, which means that it is
  at least as close to $t$ as it is to $s$. Since $s$ is clearly closer to $s$
  than to $t$, this means that the shortest path from $s$ to $v$ (that visits
  $u$ and $p$) intersects $b_{st}$ somewhere between $s$ and $p$. Since it again
  intersects $b_{st}$ at $v$, we now have a contradiction: by
  Lemma~\ref{lem:bisector_sp_intersect}, any shortest path from $s$ to $v$
  intersects $b_{st}$ at most once. The lemma follows.
\end{proof}

Let $F_s^t = e_1,..,e_g$ denote the extension segments of the vertices of
$\geod(t,s)$ and $\geod(s,z)$, ordered along $\hat{\P}(z,s,t)$, and clipped to
$\hat{\P}(z,s,t)$. See Fig.~\ref{fig:bisector_in_funnel}(b). We define $F_t^s$
analogously.

\begin{lemma}
  \label{lem:extension_segments_F}
  All vertices of $\br_{st}$ lie on clipped extension segments in
  $F_s^t \cup F_t^s$.
\end{lemma}

\begin{proof}
  By Lemma~\ref{lem:extension_segments} all vertices of $b_{st}$ in $P_r$ lie
  on $\hat{\P}(z,s,t)$. Furthermore, by Lemma~\ref{lem:bisector_in_funnel} all
  these vertices lie in $\F(z,s,t)$. Hence, it suffices to clip all extension
  segments to $\hat{\P}(z,s,t)$ (or even $\F(z,s,t)$). For all vertices on
  $\geod(t,z)$ the extension segments (with respect to $s$) are disjoint from
  $\hat{\P}(z,s,t)$. It follows that for site $s$, only the clipped extension
  segments from vertices on $\geod(s,t)$ and $\geod(s,z)$ are
  relevant. Analogously, for site $t$, only the clipped extension segments on
  $\geod(s,t)$ and $\geod(t,z)$ are relevant.
\end{proof}

\begin{observation}
  \label{obs:extension_segments_pw_disjoint}
  The extension segments in $F_s^t$ are all pairwise disjoint, start on
  $\geod(s,t)$ or $\geod(s,z)$, and end on $\geod(t,z)$.
\end{observation}


\begin{figure}[tb]
  \centering
  \includegraphics[page=3]{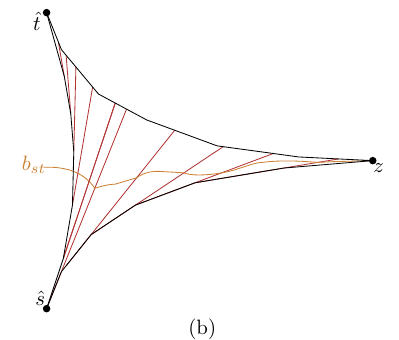}
  \caption{The extension segments in $F_s^t$ are all pairwise disjoint,
    and end at the chain from $t$ to $z$.}
  \label{fig:pseudo-triangle_slopes}
\end{figure}

By Corollary~3.29 of Aronov~\cite{aronov1989geodesic} every (clipped) extension
segment in $r \in F_s^t \cup F_t^s$ intersects $b_{st}$ (and thus $\br_{st}$)
at most once. Therefore, every such extension segment $r$ splits the bisector
in two. Together with Lemma~\ref{lem:extension_segments_F} and
Observation~\ref{obs:extension_segments_pw_disjoint} this now give us
sufficient information to efficiently binary search among the vertices of
$\br_{st}$ when we have (efficient) access to $\hat{\P}(z,s,t)$.

\begin{lemma}
  \label{lem:intersect_bst}
  Consider extension segments $e_i$ and $e_j$, with $i \leq j$, in
  $F_s^t$. If $e_i$ intersects $b_{st}$ then so does $e_j$.
\end{lemma}

\begin{proof}
  It follows from Lemma~\ref{lem:bisector_in_funnel} that $b_{st}$ intersects
  $\partial \hat{\P}(z,s,t)$ only in $z$ and in a point $w'$ on $\geod(s,t)$. Thus,
  $b_{st}$ partitions $\hat{\P}(z,s,t)$ into an $s$-side, containing $\geod(s,z)$,
  and a $t$-side, containing $\geod(t,z)$. Since the extension segments also
  partition $\hat{\P}(z,s,t)$ it then follows that the extension segments in
  $F_s^t$ intersect $b_{st}$ if and only if their starting point lies in
  the $s$-side and their ending point lies in the $t$-side. By
  Observation~\ref{obs:extension_segments_pw_disjoint} all segments in
  $F_s^t$ end on $\geod(t,z)$. Hence, they end on the $t$-side. We
  finish the proof by showing that if $e_i \in F_s^t$ starts on the
  $s$-side, so must $e_j \in F_s^t$, with $j \geq i$.

  The extension segments of vertices in $\geod(s,z)$ trivially have their start
  point on the $s$-side. It thus follows that they all intersect $b_{st}$. For
  the extension segments of vertices in $\geod(s,t)$ the ordering is such that
  the distance to $s$ is monotonically decreasing. Hence, if $e_i$ intersects
  $b_{st}$, and thus starts on the $s$-side, so does $e_j$, with $j \geq i$.
\end{proof}

\begin{lemma}
  \label{lem:fan_intersects}
  Consider extension segments $e_i$ and $e_j$, with $i \leq j$, in $F_s^t$. If
  $e_i$ intersects $\br_{st}$ then so does $e_j$.
\end{lemma}

\begin{proof}
  From Lemma~\ref{lem:intersect_bst} it follows that if $e_i$ intersects
  $b_{st}$ then so does $e_j$, with $j \geq i$. So, we only have to show that
  if $e_i$ intersects $b_{st}$ in $P_r$ then so does $e_j$. Since the
  extension segments in $F_s^t$ are pairwise disjoint, it follows that
  if $e_i$ intersects $b_{st}$, say in point $p$ then $e_j$, with $j \geq i$
  must intersect $b_{st}$ on the subcurve between $p$ and $z$. Since $b_{st}$
  intersects $d$ at most once (Lemma~\ref{lem:bst_intersect_d}), and
  $p \in P_r$, it follows that this part of the curve, and thus its
  intersection with $e_j$, also lies in $P_r$.
\end{proof}

\begin{corollary}
  \label{cor:suffix}
  The segments in $F_s^t$ that define a vertex in $\br_{st}$ form a suffix
  $G_s^t$ of $F_s^t$. That is, there is an index $a$ such that
  $G_s^t=e_a,..,e_g$ is exactly the set of extension segments in $F_s^t$ that
  define a vertex of $\br_{st}$.
\end{corollary}

When we have $\hat{\P}(z,s,t)$ and the point $w$, we can find the value $a$
from Corollary~\ref{cor:suffix} in $O(\log m)$ time as follows. We binary
search along $\geod(t,s)$ to find the first vertex $u_{a'}$ such that $u_{a'}$
is closer to $s$ then to $t$. For all vertices after $u_{a'}$, its extension
segment intersects $b_{st}$ in $\hat{\P}(z,s,t)$. To find the first segment
that intersects $b_{st}$ in $P_r$, we find the first index $a \geq a'$
for which the extension segment intersects $d$ below $w$. In total this takes
$O(\log m)$ time.

\begin{wrapfigure}[8]{r}{0.35\textwidth}
  \centering
  \vspace{-1.8\baselineskip}
  \includegraphics{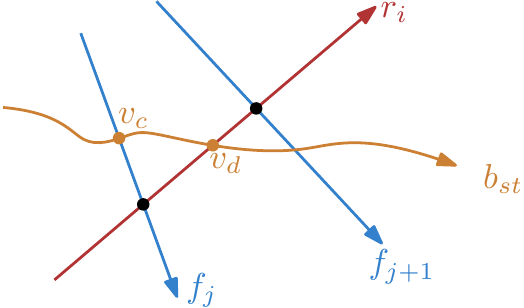}
  \caption{The bisector vertex $v_c$ on $f_j$ occurs before $v_d$ on $r_i$.}
  \label{fig:intersection_order}
\end{wrapfigure}

Let $G_s^t = r_1,..,r_{g'} = e_a,..,e_g$ be the ordered set of extension
segments that intersect $\br_{st}$. Similarly, let $G_t^s=f_1,..,f_{h'}$ be the
suffix of extension segments from $F_t^s$ that define a vertex of $\br_{st}$.

\begin{observation}
  \label{obs:intersection_order}
  Let $r_i$ be an extension segment in $G_s^t$, and let $v_d$ be the vertex of
  $b_{st}$ on $r_i$. Let $f_j$ be the last extension segment in $G_t^s$ such
  that $f_j$ intersects $r_i$ in a point closer to $s$ than to $t$. See
  Fig.~\ref{fig:intersection_order}. The vertex $v_c$ of $b_{st}$ corresponding
  to $f_j$ occurs before $v_d$, that is $c < d$.
\end{observation}

\begin{proof}
  By definition of $j$ it follows that the intersection point $v_d$ of $r_i$
  and $\br_{st}$ lies between the intersection of $r_i$ with $f_j$ and
  $f_{j+1}$. See Fig.~\ref{fig:intersection_order}. Thus, the intersection
  point
\end{proof}

\begin{lemma}
  \label{lem:random_access_bst}
  Let $j$ be the number of extension segments in $G_t^s$ that intersect
  $r_i$ in a point closer to $s$ than to $t$. Then $r_i$ contains vertex
  $v_d = v_{i+j}$ of $\br_{st}$.
\end{lemma}

\begin{proof}
  It follows from Corollary~\ref{cor:suffix} and the definition of $G_s^t$ and
  $G_t^s$ that all vertices of $\br_{st}$ lie on extension segments in
  $\{r_1,..,r_{g'}\} \cup \{f_1,..,f_{h'}\}$. Together with Corollary 3.29 of
  Aronov~\cite{aronov1989geodesic} we get that every such extension segment defines exactly
  one vertex of $\br_{st}$. Since the bisector intersects the segments
  $r_1,..,r_{g'}$ in order, there are exactly $i-1$ vertices of $\br_{st}$
  before $v_d$, defined by the extension segments in $G_s^t$. Let $f_\ell$ be
  the last extension segment in $G_t^s$ that intersects $r_i$ in a point closer
  to $s$ than to $t$. Observation~\ref{obs:intersection_order} gives us that
  this extension segment defines a vertex $v_c$ of $b_{st}$ with $c < d$. We
  then again use that $\br_{st}$ intersects the segments $f_1,..,f_{h'}$ in
  order, and thus $\ell = j$. Hence, $v_d$ is the $(i+j)^\text{th}$ vertex of
  $\br_{st}$.
\end{proof}

It follows from Lemma~\ref{lem:random_access_bst} that if we have
$\hat{\P}(z,s,t)$ and we have efficient random access to its vertices, we also
have efficient access to the vertices of the bisector $\br_{st}$. Next, we
argue with some minor augmentations the preprocessing of $P$ into a two-point
query data structure by Guibas and Hershberger gives us such access.




\paragraph{Accessing $\hat{\P}(z,s,t)$.} The data structure of Guibas and
Hershberger can return the shortest path between two query points $p$ and $q$,
represented as a balanced
tree~\cite{guibas1989query,hershberger_new_1991}. This tree is essentially a
persistent balanced search tree on the edges of the path. Every node of the
tree can access an edge $e$ of the path in constant time, and the edges are
stored in order along the path. The tree is balanced, and supports
concatenating two paths efficiently. To support random access to the vertices
of $\hat{\P}(z,s,t)$ we need two more operations: we need to be able to access
the $i^\text{th}$ edge or vertex in a path, and we need to be able to find the
longest prefix (or suffix) of a shortest path that forms a convex chain. This
last operation will allow us to find the corners $\hat{s}$ and $\hat{t}$ of
$\hat{\P}(z,s,t)$. The data structure as represented by Guibas and Hershberger
does not support these operations directly. However, with two simple
augmentations we can support them in $O(\log m)$ time. In the following, we use
the terminology as used by Guibas and Hershberger~\cite{guibas1989query}.

The geodesic between $p$ and $q$ is returned as a balanced tree. The leaves of
this tree correspond to, what Guibas and Hershberger call, \emph{fundamental
  strings}: two convex chains joined by a tangent. The individual convex chains
are stored as balanced binary search trees. The internal nodes have two
or three children, and represent \emph{derived strings}: the concatenation of
the fundamental strings stored in its descendant leaves. See
Fig.~\ref{fig:shortest_path_ds} for an illustration.

\begin{figure}[tb]
  \centering
  \includegraphics{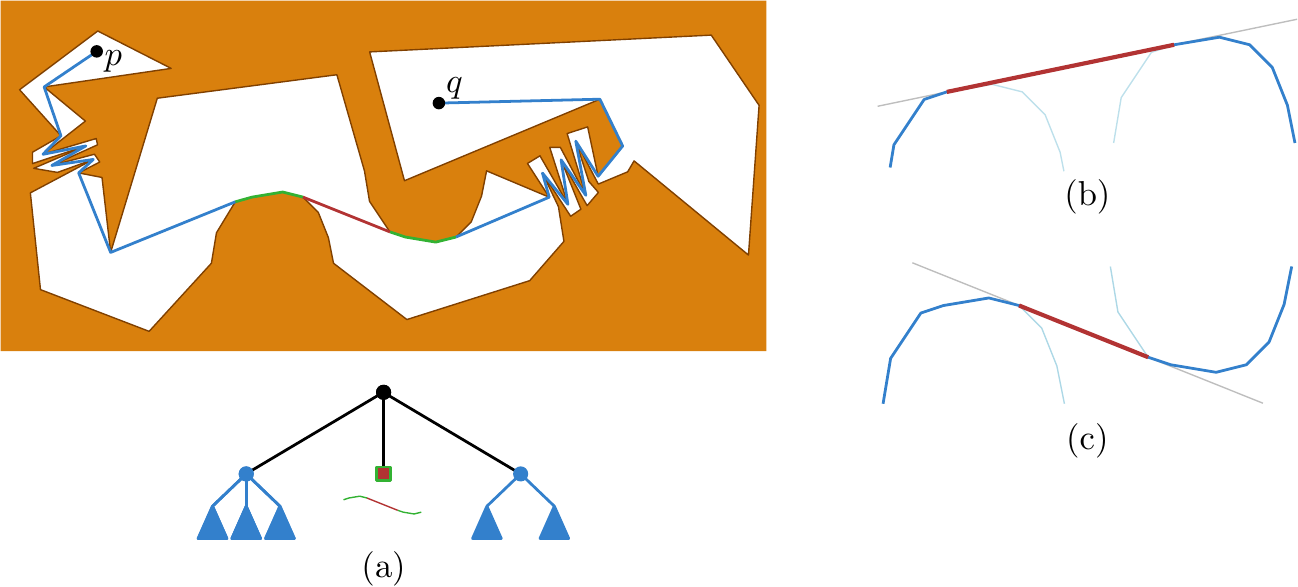}
  \caption{The data structure of Guibas and Hershberger~\cite{guibas1989query}
    can return the geodesic between two query points $p$ and $q$ as a balanced
    tree (a). The leaves of the tree correspond to fundamental strings: two
    convex chains joined by a tangent. The internal nodes represent derived
    strings: the concatenation of two or three sub-paths (strings). A
    fundamental string can be convex (b) or non-convex (c).}
  \label{fig:shortest_path_ds}
\end{figure}

To make sure that we can access the $i^\text{th}$ vertex or edge on a shortest
path in $O(\log m)$ time, we augment the trees to store subtree sizes. It is
easy to see that we can maintain these subtree sizes without affecting the
running time of the other operations.

To make sure that we can find the longest prefix (suffix) of a shortest path
that is convex we do the following. With each node $v$ in the tree we store a
boolean flag $v\mathit{.convex}$ that is true if and only if the sub path it
represents forms a convex chain. It is easy to maintain this flag without
affecting the running time of the other operations. For leaves of the tree
(fundamental strings) we can test this by checking the orientation of the
tangent with its two adjacent edges of its convex chains. These edges can be
accessed in constant time. Similarly, for internal nodes (derived strings) we
can determine if the concatenation of the shortest paths represented by its
children is convex by inspecting the $\mathit{convex}$ field of its children,
and checking the orientation of only the first and last edges of the shortest
paths. We can access these edges in constant time. This augmentation allows us
to find the last vertex $v$ of a shortest path $\geod(p,q)$ such that
$\geod(p,v)$ is a convex chain in $O(\log m)$ time. We can then obtain
$\geod(p,v)$ itself (represented by a balanced tree) in $O(\log m)$ time
by simply querying the data structure with points $p$ and $v$. Hence, we can
compute the longest prefix (or suffix) on which a shortest path forms a convex
chain in $O(\log m)$ time.


\begin{wrapfigure}[20]{r}{0.425\textwidth}
  \centering
  \vspace{-1.2\baselineskip}
  \includegraphics{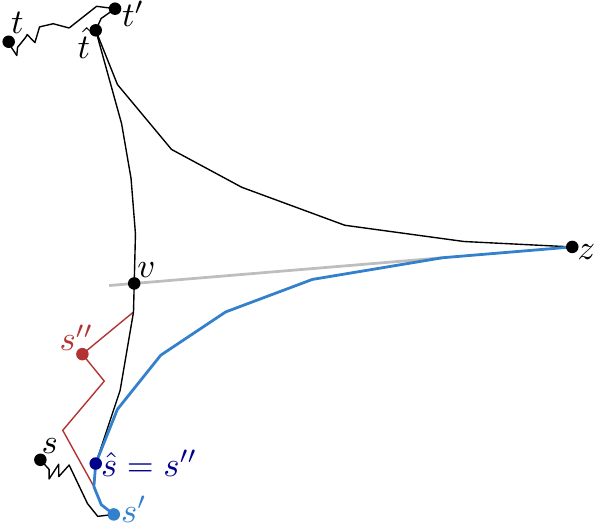}
  \vspace{-.3cm}
  \caption{We can find the first point $s'$ on $\protect\geod(s,z)$ such that
    $\protect\geod(s',z)$ is convex. When we have a point $v$ of
    $\protect\geod(s',t')$ known to be in $\hat{\P}(z,s,t)$, we can find
    $s''=\hat{s}$ (darkblue). The first point $s''$ on $\protect\geod(s',t')$
    such that $\protect\geod(s'',v)$ is convex has to lie on
    $\protect\geod(s',z)$. If this is not the case (as shown in red), then we
    can shortcut the shortest path to avoid $s''$, leading to a contradiction.
  }
  \label{fig:pseudo-triangle}
\end{wrapfigure}
Given point $z$, the above augmentation allow us to access $\hat{\P}(z,s,t)$ in
$O(\log m)$ time. We query the data structure to get the tree representing
$\geod(s,z)$, and, using our augmentations, find the longest convex suffix
$\geod(s',z)$. Similarly, we find the longest convex suffix $\geod(t',z)$ of
$\geod(t,z)$. Observe that the corners $\hat{s}$ and $\hat{t}$ of
$\hat{\P}(z,s,t)$ lie on $\geod(s',z)$ and $\geod(t',z)$, respectively
(otherwise $\geod(\hat{s},z)$ and $\geod(\hat{t},z)$ would not be convex
chains). Unfortunately, we cannot directly use the same approach to find the
part $\geod(s',t')$ that is convex, as it may both start and end with a piece
that is non-convex (with respect to $\geod(s',t')$). However, consider the
extension segment of the first edge of the shortest path from $z$ to $s$ (see
Fig.~\ref{fig:pseudo-triangle}). This extension segment intersects the shortest
path $\geod(s',t')$ exactly once in a point $v$. By construction, this point
$v$ must lie in the pseudo-triangle $\hat{\P}(z,s,t)$. Thus, we can decompose
$\geod(s',t')$ into two sub-paths, one of which starts with a convex chain and
the other ends with a convex chain. Hence, for those chains we can use the
$\mathit{convex}$ fields to find the vertices $s''$ and $t''$ such that
$\geod(s'',v)$ and $\geod(v,t'')$ are convex, and thus $\geod(s'',t'')$ is
convex. Finally, observe that $s''$ must occur on $\geod(s',z)$, otherwise we
could shortcut $\geod(s',t')$. See Fig.~\ref{fig:pseudo-triangle}. Hence, $s'' = \hat{s}$
and $t'' = \hat{t}$ are the two corners of the pseudo-triangle
$\hat{\P}(z,s,t)$. We can find $v$ in $O(\log m)$ time by a binary search on
$\geod(s',t')$. Finding the longest convex chains starting and ending in $v$
also takes $O(\log m)$ time, as does computing the shortest path
$\geod(\hat{s},\hat{t})$. It follows that given $z$, we can compute (a
representation of) $\hat{\P}(z,s,t)$ in $O(\log m)$ time.

With the above augmentations, and using Lemma~\ref{lem:random_access_bst}, we
then obtain the following result.

\begin{lemma}
  \label{lem:random_access_bst_final}
  Given the points $w$ and $z$ where $\br_{st}$ intersects $d$ and the outer
  boundary of $P_r$, respectively,
  we can access the $i^\text{th}$ vertex of $\br_{st}$ in $O(\log m)$ time.
\end{lemma}

\begin{proof}
  Recall that the data structure of Guibas and
  Hershberger~\cite{guibas1989query} reports the shortest path between query
  points $p$ and $q$ as a balanced tree. We augment these trees such that each
  node knows the size of its subtree. It is easy to do this using only constant
  extra time and space, and without affecting the other operations. We can then
  simply binary search on the subtree sizes, using
  Lemma~\ref{lem:random_access_bst} to guide the search.
\end{proof}

\paragraph{Finding $w$ and $z$.} We first show that we can find the point $w$ where
$b_{st}$ enters $P_r$ (if it exists), and then show how to find $z$, the other
point where $b_{st}$ intersects $\partial P_r$.

\begin{lemma}
  \label{lem:finding_w}
  We can find $w$ in $O(\log^2 m)$ time.
\end{lemma}

\begin{proof}
  Consider the geodesic distance of $s$ to diagonal $d$ as a function $f_s$,
  parameterized by a value $\lambda \in [0,1]$ along $d$. Similarly, let $f_t$
  be the distance function from $t$ to $d$. Since $b_{st}$ intersects $d$
  exactly once --namely in $w$-- the predicate
  $\hat{\P}(\lambda) = f_s(\lambda) < f_t(\lambda)$ changes from \textsc{True}
  to \textsc{False} (or vice versa) exactly once. Query the data structure of
  Guibas and Hershberger~\cite{guibas1989query} to get the funnel representing
  the shortest paths from $s$ to the points in $d$. Let $p_1,..,p_h$, with
  $h=O(m)$, be the intersection points of the extension segments of vertices in
  the funnel with $d$. Similarly, compute the funnel representing the shortest
  paths from $t$ to $d$. The extension segments in this funnel intersect $d$ in
  points $q_1,..,q_k$, with $k=O(m)$. We can now simultaneously binary search
  among $p_1,..,p_h$ and $q_1,..,q_k$ to find the smallest interval $I$ bounded
  by points in $\{p_1,..,p_h,q_1,..,q_k\}$ in which $\hat{\P}$ flips from
  \textsc{True} to \textsc{False}. Hence, $I$ contains $w$. Computing the
  distance from $s$ ($t$) to some $q_i$ ($p_i$) takes $O(\log m)$ time, and
  thus we can find $I$ in $O(\log^2 m)$ time. On interval $I$ both $f_s$ and
  $f_t$ are simple hyperbolic functions consisting of a single piece, and thus
  we can compute $w$ in constant time.
\end{proof}

Consider the vertices $v_1,..,v_h$ of $P_r$ in clockwise order, where
$d=\overline{v_hv_1}$ is the diagonal. Since the bisector $b_{st}$ intersects
the outer boundary of $P_r$ in only one point, there is a vertex $v_a$ such
that $v_1,..,v_a$ are all closer to $t$ than to $s$, and $v_{a+1},..,v_h$ are
all closer to $s$ than to $t$. We can thus find this vertex $v_a$ using a
binary search. This takes $O(\log^2 m)$ time, as we can compute
$\geodlen(s,v_i)$ and $\geodlen(t,v_i)$ in $O(\log m)$ time. It then follows
that $z$ lies on the edge $\overline{v_a,v_{a+1}}$. We can find the exact
location of $z$ using a similar approach as in Lemma~\ref{lem:finding_w}. This
takes $O(\log^2 m)$ time. Thus, we can find $z$ in $O(\log^2 m)$ time. We
summarize our results from this section in the following theorem.


\begin{theorem}
  \label{thm:represent_bisector}
  Let $P$ be a simple polygon with $m$ vertices that is split into $P_\ell$ and
  $P_r$ by a diagonal $d$. In $O(m)$ time we can preprocess $P$, such that for
  any pair of points $s$ and $t$ in $P_\ell$, we can compute a representation
  of $\br_{st}=b_{st} \cap P_r$ in $O(\log^2 m)$ time. This representation
  supports accessing any of its vertices in $O(\log m)$ time.
\end{theorem}

\section{Rebuilding the Forest \VD}
\label{sec:rebuilding}

Consider a level in the balanced decomposition of $P$ at which a diagonal $d$
splits a subpolygon of $P$ into $P_\ell$ and $P_r$, and assume without loss of
generality that $d$ is vertical and that $P_r$ lies right of $d$. Recall that
we partition the sites $S_\ell \subset P_\ell$ into groups (subsets). When we
insert a new site into a group $S^*$ of size $k$, or delete a site from $S^*$,
we rebuild the forest $\VD=\VD(S^*)$, representing the topology of the Voronoi
Diagram of $S^*$ in $P_r$, from scratch. We will now show that we can compute
\VD efficiently by considering it as an abstract Voronoi
diagram~\cite{klein1993avd}. Assuming that certain geometric primitives like
computing the intersections between ``related'' bisectors take $O(X)$ time we
can construct an abstract Voronoi diagram in expected $O(Xk\log k)$
time~\cite{klein1993avd}. We will show that \VD is a actually a
\emph{Hamiltonian} abstract voronoi diagram, which means that it can be
constructed in $O(Xk)$ time~\cite{klein1994hamiltonian_vd}. We show this in
Section~\ref{sub:Hamiltonian_avd}. In Section~\ref{sub:Geometric_Primitives} we
discuss the geometric primitives used by the algorithm of Klein and
Lingas~\cite{klein1994hamiltonian_vd}; essentially computing (a representation
of) the concrete Voronoi diagram of five sites. We show that we can implement
these primitives in $O(\log^2 m)$ time by computing the intersection point
between two ``related'' bisectors $\br_{st}$ and $\br_{tu}$. This then gives us
an $O(k\log^2 m)$ time algorithm for constructing \VD. Finally, in
Section~\ref{sub:query} we argue that having only the topological structure \VD
is sufficient to find the site in $S^*$ closest to a query point
$q \in P_r$.

\subsection{Hamiltonian abstract Voronoi Diagrams}
\label{sub:Hamiltonian_avd}

In this section we show that we can consider $\VD$ as a Hamiltonian abstract
Voronoi diagram. A Voronoi diagram is \emph{Hamiltonian} if there is a curve
--in our case the diagonal $d$-- that intersects all regions exactly once, and
furthermore this holds for all subsets of the
sites~\cite{klein1994hamiltonian_vd}. Let $S^*$ be the set of sites in $P_\ell$
that we consider, and let $T^*$ be the subset of sites from $S^*$ whose Voronoi
regions intersect $d$, and thus occur in \VD.

\begin{lemma}
  \label{lem:hamiltonian_avd}
  The Voronoi diagram $\VD(T^*)$ in $P_r$ is a Hamiltonian abstract Voronoi diagram.
\end{lemma}

\begin{proof}
  By Lemma~\ref{lem:bst_intersect_d} any bisector $b_{st}$ intersects the
  diagonal $d$ at most once. This implies that for any subset of sites
  $T \subseteq S^*$, so in particular for $T^*$, the diagonal $d$ intersects
  all Voronoi regions in $\VD(T)$ at most once. By definition, $d$ intersects
  all Voronoi regions of the sites in $T^*$ at least once. What remains is to
  show that this holds for any subset of $T^*$.  This follows since the Voronoi
  region $V(s,T_1 \cup T_2)$ of a site $s$ with respect to a set $T_1 \cup T_2$
  is contained in the voronoi region $V(s,T_1)$ of $s$ with respect to $T_1$.
\end{proof}

\paragraph{Computing the Order Along $d$.} We will use the algorithm of Klein
and Lingas~\cite{klein1994hamiltonian_vd} to construct
$\VD=\VD(S^*)=\VD(T^*)$. To this end, we need the set of sites $T^*$ whose
Voronoi regions intersect $d$, and the order in which they do so. Next, we show
that we can maintain the sites in $S^*$ so that we can compute this information
in $O(k\log^2 m)$ time.


\begin{lemma}
  \label{lem:order}
  Let $s_1,..,s_k$ denote the sites in $S^*$ by increasing distance from the
  bottom-endpoint $p$ of $d$, and let $t_1,..,t_z$ be the subset $T^*
  \subseteq S^*$ of sites whose Voronoi regions intersect $d$, ordered along $d$
  from bottom to top. For any pair of sites $t_a=s_i$ and $t_c=s_j$, with
  $a < c$, we have that $i < j$.
\end{lemma}

\begin{proof}
  Since $t_a$ and $t_c$ both contribute Voronoi regions intersecting $d$,
  their bisector must intersect $d$ in some point $w$ in between these two
  regions. Since $a < c$ it then follows that all points on $d$ below $w$,
  so in particular the bottom endpoint $p$, are closer to $t_a=s_i$ than to
  $t_c = s_j$. Thus, $i < j$.
\end{proof}

Lemma~\ref{lem:order} suggests a simple iterative algorithm for
extracting $T^*$ from $S^*=s_1,..,s_k$.

\begin{lemma}
  \label{lem:compute_sequence}
  Given $S^*=s_1,..,s_k$, we can compute $T^*$ from $S^*$ in $O(k\log^2 m)$
  time.
\end{lemma}




\begin{proof}
  We consider the sites in $S^*$ in increasing order, while maintaining $T^*$
  as a stack. More specifically, we maintain the invariant that when we start
  to process $s_{j+1}$, $T^*$ contains exactly those sites among $s_1,..,s_j$
  whose Voronoi region intersects $d$, in order along $d$ from bottom to top.

  Let $s_{j+1}$ be the next site that we consider, and let $t=s_i$, for some
  $i \leq j$, be the site currently at the top of the stack. We now compute the
  distance $\geod(s_{j+1},q)$ between $s_{j+1}$ and the topmost endpoint $q$ of
  $d$. If this distance is larger than $\geod(t,q)$, it follows that the
  Voronoi region of $s_{j+1}$ does not intersect $d$: since the bottom endpoint
  $p$ of $d$ is also closer to $t=s_i$ than to $s_{j+1}$, all points on $d$ are
  closer to $t$ than to $s_{j+1}$.

  If $\geod(s_{j+1},q)$ is at most $\geod(t,q)$ then the Voronoi region of
  $s_{j+1}$ intersects $d$ (since $t$ was the site among $s_1,..,s_j$ that was
  closest to $q$ before). Furthermore, since $p$ is closer to $t=s_i$ than to
  $s_{j+1}$ the bisector between $s_{j+1}$ and $t$ must intersect $d$ in some
  point $a$. If this point $a$ lies above the intersection point $c$ of $d$
  with the bisector between $t$ and the second site $t'$ on the stack, we have
  found a new additional site whose Voronoi region intersects $d$. We push
  $s_{j+1}$ onto $T^*$ and continue with the next site $s_{j+2}$. Note that the
  Voronoi region of every site intersects $d$ in a single segment, and thus
  $T^*$ correctly represents all sites intersecting $d$. If $a$ lies below $c$
  then the Voronoi region of $t$, with respect to $s_1,..s_{j+1}$, does not
  intersect $d$. We thus pop $t$ from $T^*$, and repeat the above procedure,
  now with $t'$ at the top of the stack.

  Since every site is added to and deleted from $T^*$ at most once the
  algorithm takes a total of $O(k)$ steps. Computing $\geod(s_{j+1},q)$ takes
  $O(\log m)$ time, and finding the intersection between $d$ and the bisector
  of $s_{j+1}$ and $t$ takes $O(\log^2 m)$ time
  (Lemma~\ref{lem:finding_w}). The lemma follows.
\end{proof}

We now simply maintain the sites in $S^*$ in a balanced binary search tree on
increasing distance to the bottom endpoint $p$ of $d$. It is easy to maintain
this order in $O(\log m + \log k)$ time per upate. We then extract the set of sites $T^*$
that have a Voronoi region intersecting $d$, and thus $P_r$, ordered along $d$
using the algorithm from Lemma~\ref{lem:compute_sequence}.

\subsection{Implementing the Required Geometric Primitives}
\label{sub:Geometric_Primitives}

In this section we discuss how to implement the geometric primitives needed by
the algorithm of Klein and Lingas~\cite{klein1994hamiltonian_vd}. 
They describe their algorithm in terms of the following two basic operations:
(i) compute the concrete Voronoi diagram of five sites, and (ii) insert a new
site $s$ into the existing Voronoi diagram $\VD(S)$. In their analysis, this
first operation takes constant time, and the second operation takes time
proportional to the size of $\VD(S)$ that lies inside the Voronoi region of $t$
in $\VD(S\cup\{t\})$. We observe that that to implement these operations it is
sufficient to be able to compute the intersection between two ``related''
bisectors $b_{st}$ and $b_{tu}$ --essentially computing the Voronoi diagram of
three sites-- and to test if a given point $q$ lies on the $s$-side of the
bisector $b_{st}$ (i.e.~testing if $q$ is ``closer to'' $s$ than to $t$). We
then show that in our setting we can implement these operations in
$O(\log^2 m)$ time, thus leading to an $O(k\log^2 m)$ time algorithm to compute
$\VD$.

\begin{observation}
  \label{obs:intersect_once}
  Let $S$ be a set of sites whose Voronoi diagram in domain $D$ is Hamiltonian,
  and let $s$, $t$, and $u$ be three sites in $S$. The bisectors $b_{st}$ and
  $b_{tu}$ intersect at most once inside $D$.
\end{observation}

\begin{proof}
  Consider an intersection point $p$ between $b_{st}$ and $b_{tu}$ in $D$. This
  point also appears on $b_{su}$ (see
  e.g.~\cite{klein1989avd,klein1994hamiltonian_vd}), and thus it appears as a
  degree three vertex in the Voronoi diagram $\VD(\{s,t,u\})$. Since
  $\VD(\{s,t,u\})$ is also Hamiltonian, it is a forest that partitions the
  domain $D$ into three simply connected regions: one for each site. It follows
  that there can be at most one degree three vertex in $\VD(\{s,t,u\})$,
  otherwise we get more than three regions in $D$.
\end{proof}

\begin{wrapfigure}[25]{r}{0.28\textwidth}
  \vspace{-2\baselineskip}
  \centering
  \includegraphics{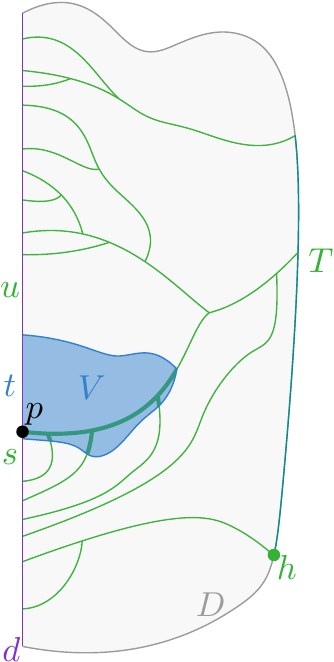}
  \caption{The part of $T$, the tree representing the Hamiltonian Voronoi
    diagram (green), that lies inside the Voronoi region $V$ (blue) of a new
    site $t$ is a subtree $T'$ (fat). We can compute $T'$, by exploring $T$
    from a point $p$ inside $V$. In case $t$ is the first site in the ordering
    along $d$ we can start from the root $h$ of the ``first'' tree in \VD.
  }
  \label{fig:insert_avd}
\end{wrapfigure}

\paragraph{Inserting a new site.} Klein and
Lingas~\cite{klein1994hamiltonian_vd} sketch the following algorithm to insert
a new site $t$ into the Hamiltonian Voronoi diagram $\VD(S)$ of a set of sites
$S$. We provide some missing details of this procedure, and briefly argue that
we can use it to insert into a diagram of three sites. Let $D$ denote the
domain in which we are interested in $\VD(S)$ (in our application, $D$ is the
subpolygon $P_r$) and let $d$ be the curve that intersects all regions in
$\VD(S)$. Recall that $\VD(S)$ is a forest. We root all trees such that the
leaves correspond to the intersections of the bisectors with $d$. The roots of
the trees now corresponds to points along the boundary $\partial D$ of $D$. We
connect them into one tree $T$ using curves along $\partial D$. Now consider
the Voronoi region $V$ of $t$ with respect to $S\cup\{t\}$, and observe that
$T \cap V$ is a subtree $T'$ of $T$. Therefore, if we have a starting point $p$
on $T$ that is known to lie in $V$ (and thus in $T'$), we can compute $T'$
simply by exploring $T$. To obtain $\VD(S \cup \{t\})$ we then simply remove
$T'$, and connect up the tree appropriately. See Fig.~\ref{fig:insert_avd} for
an illustration. We can test if a vertex $v$ of $T$ is part of $T'$ simply by
testing if $v$ lies on the $t$-side of the bisector between $t$ and one of
the sites defining $v$. We can find the exact point $q$ where an edge
$(u,v)$ of $T$, representing a piece of a bisector $b_{su}$ leaves $V$ by
computing the intersection point of $b_{su}$ with $b_{tu}$ and $b_{st}$.

We can find the starting point $p$ by considering the order of the Voronoi
regions along $d$. Let $s$ and $u$ be the predecessor and successor of $t$ in
this order. Then the intersection point of $d$ with $b_{su}$ must lie in
$V$. This point corresponds to a leaf in $T$. In case $t$ is the first site in
the ordering along $d$ we start from the root $h$ of the tree that contains the
bisector between the first two sites in the ordering; if this point is not on
the $t$-side of the bisector between $t$ and one of the sites defining $h$ then
$b_{tu}$ forms its own tree (which we then connect to the global root). We do
the same when $t$ is the last point in the ordering. This procedure requires
$O(|T'|)$ time in total (excluding the time it takes to find $t$ in the
ordering of $S$; we already have this information when the procedure is used in
the algorithm of Klein and Lingas~\cite{klein1994hamiltonian_vd}).

We use the above procedure to compute the Voronoi diagram of five sites in
constant time: simply pick three of the sites $s$, $t$, and $u$, ordered along
$d$, compute their Voronoi diagram by computing the intersection of $b_{st}$
and $b_{tu}$ (if it exists), and insert the remaining two sites. Since the
intermediate Voronoi diagrams have constant size, this takes constant time.

\begin{wrapfigure}[8]{r}{0.45\textwidth}
  \vspace{-2.1\baselineskip}
  \centering
  \includegraphics{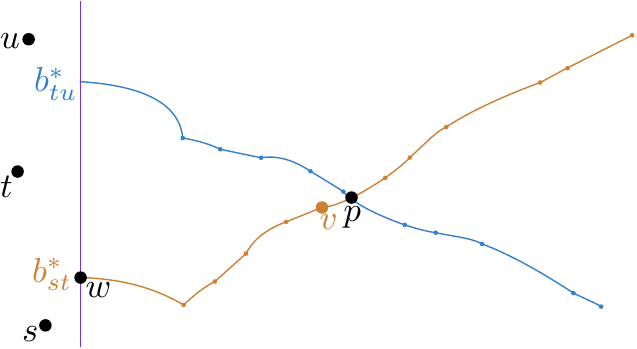}
  \caption{We find the intersection point of the two bisectors by binary
    searching along $\br_{st}$. }
  \label{fig:find_intersection_binsearch}
\end{wrapfigure}
\paragraph{Computing the intersection of bisectors $\br_{st}$ and $\br_{tu}$. }
Since \VD is a Hamiltonian Voronoi diagram, any any pair of bisectors
$\br_{st}$ and $\br_{tu}$, with $s,t,u \in T^*$, intersect at most once
(Observation~\ref{obs:intersect_once}). Next, we show how to compute
this intersection point (if it exists).




\begin{lemma}
  \label{lem:binary_search_funnel}
  Given $\hat{\P}(z,s,t)$ and $\hat{\P}(z',t,u)$, we can find the
  intersection point $p$ of $\br_{st}$ and $\br_{tu}$ (if it exists) in
  $O(\log^2 m)$ time.
\end{lemma}

\noindent{\textit{Proof.} }
  We will find the edge of $\br_{st}$ containing the intersection point $p$ by
  binary searching along the vertices of $\br_{st}$. Analogously we find the
  edge of $\br_{tu}$ containing $p$. It is then easy to compute the exact
  location of $p$ in constant time.

  Let $w$ be the starting point of $\br_{st}$, i.e.~the intersection of
  $b_{st}$ with $d$, and assume that $t$ is closer to $w$ than $u$, that is,
  $\geodlen(t,w) < \geodlen(u,w)$ (the other case is symmetric). In our binary
  search, we now simply find the last vertex $v=v_k$ for which
  $\geodlen(t,v) < \geodlen(u,v)$. It then follows that $p$ lies on the edge
  $(v_k,v_{k+1})$ of $\br_{st}$. See
  Fig.~\ref{fig:find_intersection_binsearch}. Using
  Lemma~\ref{lem:random_access_bst_final} we can access any vertex of
  $\br_{st}$ in $O(\log m)$ time. Thus, this procedure takes $O(\log^2 m)$ time
  in total.  \hfill\qed

Note that we can easily extend the algorithm from
Lemma~\ref{lem:binary_search_funnel} to also return the actual edges of
$\br_{st}$ and $\br_{tu}$ that intersect. With this information we can
construct the cyclic order of the edges incident to the vertex of \VD
representing this intersection point. It now follows that for every group $S^*$
of sites in $P_\ell$, we can compute a representation of \VD of size $O(k)$ in
$O(k\log^2 m)$ time.



\subsection{Planar Point Location in \VD}
\label{sub:query}

In this section we show that we can efficiently answer point location queries,
and thus nearest neighbor queries using \VD.


\begin{lemma}
  \label{lem:bisectors_x-monotone}
  For $s,t \in S_\ell$, the part of the bisector $\br_{st} = b_{st} \cap P_r$
  that lies in $P_r$ is $x$-monotone.
\end{lemma}

  \begin{wrapfigure}[14]{r}{0.31\textwidth}
    \vspace{-.5cm}
    \centering
    \includegraphics{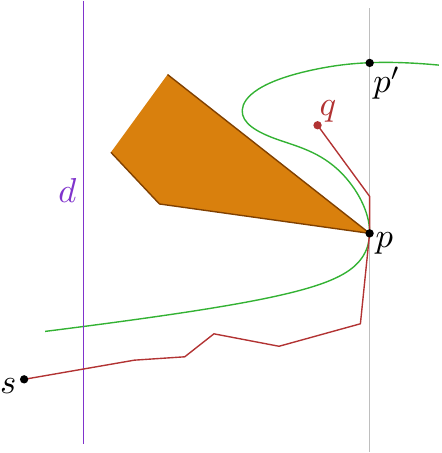}
    \caption{A non $x$-monotone bisector can occur only in \mbox{degenerate} inputs.}
    \label{fig:non_x-monotone}
  \end{wrapfigure}

\noindent\textit{Proof.}
  Assume, by contradiction, that $b_{st}$ is not $x$-monotone in $P_r$, and let
  $p$ be a point on $b_{st}$ such that $p_x$ is a local maximum. Since
  $\br_{st}$ is not $x$-monotone, it intersects the vertical line through $p$
  also in another point $p'$ further along $\br_{st}$. Let $q$ be a point in
  the region enclosed by the subcurve along $\br_{st}$ from $p$ to $p'$ and
  $\overline{pp'}$. See Fig.~\ref{fig:non_x-monotone}. This means that either
  $\geod(s,q)$ or $\geod(t,q)$ is non $x$-monotone. Assume without loss of
  generality that it is $\geod(s,q)$. It is now easy to show that $\geod(s,q)$
  must pass through $p$. However, that means that $\br_{st}$ (and thus
  $b_{st}$) touches the polygon boundary in $p$. By the general position
  assumption $b_{st}$ has no points in common with $\partial P$ other than its
  end points. Contradiction.
\hfill\qed

Since the (restriction of the) bisectors are $x$-monotone
(Lemma~\ref{lem:bisectors_x-monotone}) we can preprocess \VD for point location
using the data structure of Edelsbrunner and
Stolfi~\cite{edelsbrunner_optimal_1986}. Given the combinatorial embedding of
\VD, this takes $O(|\VD|)$ time. To decide if a query point $q$ lies above or
below an edge $e \in \VD$ we simply compute the distances $\geodlen(s,q)$ and
$\geodlen(t,q)$ between $q$ and the sites $s$ and $t$ defining the bisector
corresponding to edge $e$. This takes $O(\log m)$ time. Point $q$ lies on the
side of the site that has the shorter distance. It follows that we can
preprocess \VD in $O(k)$ time, and locate the Voronoi region containing a query
point $q$ in $O(\log k\log m)$ time. We summarize our results in the following Lemma.

\begin{lemma}
  \label{lem:reconstruct_v}
  Given a set of $k$ sites $S^*$ in $P_\ell$, ordered by increasing distance
  from the bottom-endpoint of $d$, we can construct the forest \VD representing
  the Voronoi diagram of $S^*$ in $P_r$ in $O(k\log^2 m)$ time. Given \VD, we
  can find the site $s \in S^*$ closest to a query point $q \in P_r$ in
  $O(\log k\log m)$ time.
\end{lemma}

\section{Putting Everything Together}
\label{sec:Together}

For every level of the balanced decomposition, we split $S_\ell$ into
$O(\sqrt{n})$ groups, and for each group build the Voronoi diagram in $P_r$
using the approach from Section~\ref{sec:rebuilding}. We process the sites in
$S_r$ analogously. This leads to the following result.

\begin{theorem}
  \label{thm:nn_search}
  Given a polygon $P$ with $m$ vertices, we can build a fully dynamic data
  structure of size $O(n\log m + m)$ that maintains a set of $n$ point sites in
  $P$ and allows for geodesic nearest neighbor queries in
  $O(\sqrt{n}\log n\log^2 m)$ time. Inserting or deleting a site takes
  $O(\sqrt{n}\log^3 m)$ time.
\end{theorem}

\begin{proof}
  Since every site occurs in exactly one group per level of the balanced
  decomposition, and there are $O(\log m)$ levels, the total space required for
  our data structure is $O(n\log m + m)$. Similarly, when we insert or delete a
  site, we have to insert or delete it in $O(\log m)$ levels. Each such
  insertion or deletion takes $O(\sqrt{n}\log^2 m)$ time, as we have to rebuild
  a Voronoi diagram for one group. For a query, we have to query all
  $O(\sqrt{n})$ groups in $O(\log m)$ levels. This leads to a total query time
  of $O(\sqrt{n}\log n\log^2 m)$.
\end{proof}

\paragraph{Insertions only.} In case our data structure has to support only
insertions, we can improve the insertion time to $O(\log n\log^3 m)$, albeit
being amortized. Recall that at every node of the balanced decomposition
our data structure partitions the sites in $S_\ell$ in $O(\sqrt{n})$ groups of
size $O(\sqrt{n})$ each. Instead, we now maintain these sites in groups of size
$2^i$, for $i \in [1..O(\log n)]$. When we insert a new site, we may get two
groups of size $2^i$. We then remove these groups, and construct a new group of
size $2^{i+1}$. For this group we rebuild the Voronoi diagram $\VD$ that these
sites induce on $P_r$ from scratch. 
Using a standard binary counter argument it can be shown that every data
structure of size $2^{i+1}$ gets rebuild (at a cost of $O(2^i\log^2 m)$) only
after $2^i$ new sites have been inserted~\cite{overmars1983design}. So, if we
charge $O(\log n\log^2 m)$ to each site, it can pay for rebuilding all of the
structures it participates in. We do this for all $O(\log m)$ levels in the
balanced decomposition, hence we obtain the following result.

\begin{corollary}
  \label{cor:insert_only}
  Given a polygon $P$ with $m$ vertices, we can build an insertion-only data
  structure of size $O(n\log m + m)$ that stores a set of $n$ point sites in
  $P$, and allows for geodesic nearest neighbor queries in worst-case
  $O(\log^2 n\log^2 m)$ time, and allows inserting a site in amortized
  $O(\log n\log^3 m)$ time.
\end{corollary}

\paragraph{Offline updates.} When we have both insertions and deletions, but the
order of these operations is known in advance, we can maintain $S$ in amortized
$O(\log n\log^3 m)$ time per update, where $n$ is the maximum number of sites
in $S$ at any particular time. Queries take $O(\log^2 n\log^2 m)$ time, and may
arbitrarily interleave with the updates. Furthermore, we do not have to know
them in advance.

For ease of description, we assume that the total number of updates $N$ is
proportional to the number of sites at any particular time, i.e.~$N \in
O(n)$. We can easily extend our approach to larger $N$ by grouping the updates
in $N/n$ groups of size $O(n)$ each. Consider a node of the balanced
decomposition whose diagonal that splits its subpolygon into $P_\ell$ and
$P_r$. We partition the sites in $S_\ell$ into groups such that at any time, a
query $q \in P_r$ can be answered by considering the Voronoi diagrams in $P_r$
of only $O(\log n)$ groups. We achieve this by building a segment
tree on the intervals during which the sites are
``alive''. More specifically, let $[t_1,t_2]$ denote a time interval in which
a site $s$ should occur in $S_\ell$ (i.e.~$s$ lies in $P_\ell$ and there is an
\acall{Insert}{$s$} operation at time $t_1$ and its corresponding
\acall{Delete}{$s$} at time $t_2$). We store the intervals of all sites in
$S_\ell$ in a segment tree~\cite{bkos2008}. Each node $v$ in this tree is
associated with a subset $S_v$ of the sites from $S_\ell$. We build the Voronoi
diagram that $S_v$ induces on $P_r$. Every site occurs in $O(\log n)$ subsets,
and in $O(\log m)$ levels of the balanced decomposition, so the total size of
our data structure is $O(n\log n\log m + m)$. Building the Voronoi diagram for
each node $v$ takes $O(|S_v|\log^2 m)$ time. Summing these results over all
nodes in the tree, and all levels of the balanced decomposition, the total
construction time is $O(n\log n\log^3 m + m)$ time. We conclude:

\begin{corollary}
  \label{cor:offline}
  Given a polygon $P$ with $m$ vertices, and a sequence of operations that
  either insert a point site inside $P$ into a set $S$, or delete a site from
  $S$, we can maintain a dynamic data structure of size $O(n\log n\log m + m)$,
  where $n$ is the maximum number of sites in $S$ at any time, that stores $S$,
  and allows for geodesic nearest neighbor queries in $O(\log^2 n\log^2 m)$
  time. Updates take amortized $O(\log n\log^3 m)$ time.
\end{corollary}


\section*{Acknowledgments}

We would like to thank Constantinos Tsirogiannis for suggesting the migration
problem that originally started this research, and Luis Barba for useful
discussions about dynamic Voronoi diagrams. This work was supported by the
Danish National Research Foundation under grant nr.~DNRF84.


\bibliographystyle{abbrv}
\bibliography{geodesic_voronoi_diagram}

\begin{thebibliography}{10}

\bibitem{agarwal1995dynamic}
P.~K. Agarwal and J.~Matou{\v{s}}ek.
\newblock {Dynamic Half-Space Range Reporting and its Applications}.
\newblock {\em Algorithmica}, 13(4):325--345, 1995.

\bibitem{aronov1989geodesic}
B.~Aronov.
\newblock {On the Geodesic Voronoi Diagram of Point Sites in a Simple Polygon}.
\newblock {\em Algorithmica}, 4(1):109--140, 1989.

\bibitem{bentley1980decomposable}
J.~L. Bentley and J.~B. Saxe.
\newblock Decomposable searching problems {I}. {Static-to-dynamic}
  transformation.
\newblock {\em Journal of Algorithms}, 1(4):301--358, 1980.

\bibitem{burrows2014geographical}
M.~T. Burrows, D.~S. Schoeman, A.~J. Richardson, J.~G. Molinos, A.~Hoffmann,
  L.~B. Buckley, P.~J. Moore, C.~J. Brown, J.~F. Bruno, C.~M. Duarte, et~al.
\newblock Geographical limits to species-range shifts are suggested by climate
  velocity.
\newblock {\em Nature}, 507(7493):492--495, 2014.

\bibitem{chan2010dynamic_ch}
T.~M. Chan.
\newblock {A Dynamic Data Structure for 3-D Convex Hulls and 2-D Nearest
  Neighbor Queries}.
\newblock {\em J. ACM}, 57(3):16:1--16:15, Mar. 2010.

\bibitem{bkos2008}
M.~de~Berg, O.~Cheong, M.~van Kreveld, and M.~Overmars.
\newblock {\em {Computational Geometry: Algorithms and Applications}}.
\newblock Springer, 3rd edition, 2008.

\bibitem{dobkin1991maintenance}
D.~Dobkin and S.~Suri.
\newblock {Maintenance of Geometric Extrema}.
\newblock {\em J. ACM}, 38(2):275--298, Apr. 1991.

\bibitem{edelsbrunner_optimal_1986}
H.~Edelsbrunner, L.~Guibas, and J.~Stolfi.
\newblock Optimal {Point} {Location} in a {Monotone} {Subdivision}.
\newblock {\em SIAM Journal on Computing}, 15(2):317--340, May 1986.

\bibitem{guibas1987balanced}
L.~Guibas, J.~Hershberger, D.~Leven, M.~Sharir, and R.~E. Tarjan.
\newblock {Linear-Time Algorithms for Visibility and Shortest Path Problems
  Inside Triangulated Simple Polygons}.
\newblock {\em Algorithmica}, 2(1):209--233, 1987.

\bibitem{guibas1989query}
L.~J. Guibas and J.~Hershberger.
\newblock {Optimal Shortest Path Queries in a Simple Polygon}.
\newblock {\em Journal of Computer and System Sciences}, 39(2):126 -- 152,
  1989.

\bibitem{hershberger_new_1991}
J.~Hershberger.
\newblock A new data structure for shortest path queries in a simple polygon.
\newblock {\em Information Processing Letters}, 38(5):231--235, June 1991.

\bibitem{hershberger1999sssp}
J.~Hershberger and S.~Suri.
\newblock {An Optimal Algorithm for Euclidean Shortest Paths in the Plane}.
\newblock {\em SIAM J. Comput.}, 28(6):2215--2256, 1999.

\bibitem{dynamic2017kaplan}
H.~Kaplan, W.~Mulzer, L.~Roditty, P.~Seiferth, and M.~Sharir.
\newblock {Dynamic Planar Voronoi Diagrams for General Distance Functions and
  their Algorithmic Applications}.
\newblock In {\em Proc. 28th Annual ACM-SIAM Symposium on Discrete Algorithms}.
  SIAM, 2017.

\bibitem{klein1989avd}
R.~Klein.
\newblock {\em {Concrete and abstract Voronoi diagrams}}, volume 400.
\newblock Springer Science \& Business Media, 1989.

\bibitem{klein1994hamiltonian_vd}
R.~Klein and A.~Lingas.
\newblock {\em {Hamiltonian abstract Voronoi diagrams in linear time}}, pages
  11--19.
\newblock Springer Berlin Heidelberg, Berlin, Heidelberg, 1994.

\bibitem{klein1993avd}
R.~Klein, K.~Mehlhorn, and S.~Meiser.
\newblock {Randomized incremental construction of abstract Voronoi diagrams}.
\newblock {\em Computational Geometry}, 3(3):157 -- 184, 1993.

\bibitem{oh_ahn2017voronoi}
E.~Oh and H.-K. Ahn.
\newblock {Voronoi Diagrams for a Moderate-Sized Point-Set in a Simple
  Polygon}.
\newblock In B.~Aronov and M.~J. Katz, editors, {\em 33rd International
  Symposium on Computational Geometry (SoCG 2017)}, volume~77 of {\em Leibniz
  International Proceedings in Informatics (LIPIcs)}, pages 52:1--52:15,
  Dagstuhl, Germany, 2017. Schloss Dagstuhl--Leibniz-Zentrum fuer Informatik.

\bibitem{ordonez2013climatic}
A.~Ordonez and J.~W. Williams.
\newblock Climatic and biotic velocities for woody taxa distributions over the
  last 16 000 years in eastern north america.
\newblock {\em Ecology letters}, 16(6):773--781, 2013.

\bibitem{overmars1983design}
M.~H. Overmars.
\newblock {\em The design of dynamic data structures}, volume 156.
\newblock Springer Science \& Business Media, 1983.

\bibitem{papadopoulou1998geodesic}
E.~Papadopoulou and T.~D. Lee.
\newblock {A New Approach for the Geodesic Voronoi Diagram of Points in a
  Simple Polygon and Other Restricted Polygonal Domains}.
\newblock {\em Algorithmica}, 20(4):319--352, 1998.

\end{thebibliography}

\end{document}